\documentclass[
times,
aps,
twocolumn,
showpacs,
superscriptaddress,
pre,
floatfix]{revtex4-1}

\usepackage{graphicx}
\usepackage{array}
\usepackage{color}
\usepackage{upgreek}
\usepackage{float}
\usepackage{enumerate}
\usepackage{enumitem}
\usepackage{lipsum}
\usepackage{booktabs}
\usepackage{url}
\usepackage{braket}

\usepackage{amsthm}
\usepackage{bm}
\usepackage[T1]{fontenc}
\usepackage{scrextend}
\usepackage{hyperref}
\usepackage[dvipsnames]{xcolor}
\definecolor{C1}{RGB}{52, 89, 149}
\definecolor{C2}{RGB}{251, 77, 61}
\definecolor{C3}{RGB}{3, 206, 164}
\definecolor{C4}{RGB}{202, 21, 81}
\hypersetup{colorlinks=true, linkcolor=C2, citecolor=C2, urlcolor=C2}

\usepackage{dsfont}
\usepackage{epstopdf}
\usepackage{tikz}
\usepackage{mathrsfs}
\usepackage[export]{adjustbox}
\usepackage{thmtools}
\usepackage{thm-restate}

\usepackage{fdsymbol}

\usepackage{accents}

\newtheorem{thm}{Theorem}

\newtheorem{obs}[thm]{Observation}

\theoremstyle{remark}

\newcommand*{\ups}{\Upsilon}
\newcommand*{\ot}{\otimes}
\newcommand*{\nn}{\nonumber}

\newcommand*{\id}{\mathds{1}}

\newcommand*{\mc}{\mathcal}
\newcommand*{\dg}{\dagger}
\newcommand*{\ex}{\mathrm{e}}

\DeclareMathOperator{\tr}{tr}

\hypersetup{
pdftitle={Scrambling is Necessary but Not Sufficient for Chaos},
pdfsubject={Many-body quantum chaos},
pdfauthor={Neil Dowling, Pavel Kos, and Kavan Modi},
pdfkeywords={Quantum chaos, Many-body quantum physics, OTOC}
}

\begin{document}
\title[]{Scrambling is Necessary but Not Sufficient for Chaos} 

\author{Neil Dowling}
\email[]{neil.dowling@monash.edu}
\address{School of Physics \& Astronomy, Monash University, Clayton, VIC 3800, Australia}

\author{Pavel Kos}
\affiliation{Max-Planck-Institut f\"ur Quantenoptik, Hans-Kopfermann-Str. 1, 85748 Garching}

\author{Kavan Modi}
\address{School of Physics \& Astronomy, Monash University, Clayton, VIC 3800, Australia}
\affiliation{Centre for Quantum Technology, Transport for New South Wales, Sydney, NSW 2000, Australia}


\begin{abstract}
    We show that out-of-time-order correlators (OTOCs) constitute a probe for Local-Operator Entanglement (LOE). There is strong evidence that a volumetric growth of LOE is a faithful dynamical {indicator} of quantum chaos, while OTOC decay corresponds to operator scrambling, often conflated with chaos. We show that rapid OTOC decay is a necessary but not sufficient condition for linear (chaotic) growth of the LOE entropy. We analytically support our results through wide classes of local-circuit models of many-body dynamics, including both integrable and non-integrable dual-unitary circuits. We show sufficient conditions under which local dynamics leads to an equivalence of scrambling and chaos.  
\end{abstract}


\maketitle

\textit{Introduction.---} 
{The question of quantum chaos is a long-standing issue. In recent years, a wide plethora of apparently inequivalent notions of quantum chaos have appeared~\cite{BerryTabor1977,bgs1984,kos2018many,amos2018,Haake2018-cs,Prosen2007,Prosen2007a,Parker2019,Yan2020-ae,Kos2020,Oliviero2021-rf,Leone2021-ov,reichl2021transition,Anand2021-yi,Rosa2022,Dowling2022}. Among them, the most well-known defining features of quantum chaotic models are universal spectral fluctuations, which match those of random matrix theory. They were shown to arise for systems with chaotic semiclassical limits~\cite{BerryTabor1977,bgs1984}. In the absence of this limit, universal fluctuations were subsequently used as a definition of chaos and were recently extensively investigated in many-body settings~\cite{Dubertrand2016,kos2018many,amos2018,amos2018v2,Bertini2018KI,Bertini2021random, Garratt2021local, suntajs2020, braun2020transition, Srdinsek2021, Delacretaz2023, Cotler2017, saad2019semiclassical,Winer2020,bertini2022exact}.

Nevertheless, the spectral definition of quantum chaos comes with a few limitations. Firstly, taking the thermodynamic limit is non-trivial, as the discrete spectrum becomes hard to treat and is difficult to access experimentally and analytically. The spectrum might not even be defined, for instance in time-dependent evolution.
Therefore, a dynamical indicator of chaos, which is well-defined for infinite systems and finite times, is clearly an attractive prospect. 
 A popular candidate is the out-of-time-ordered correlator (OTOC), which measures a notion of scrambling in many-body systems~\cite{Shenker_Stanford_2014,Maldacena_Shenker_Stanford_2016,Swingle2016,Roberts2016,Swingle2018,Foini2019,Sunderhauf2019,swingle2020,Xu2022-ue}.
But its definition of `chaos' does not agree with the spectral one~\cite{pappalardi2018,Hashimoto2020,Pilatowsky2020,Xu2020,Shor_2021,Balachandran2021-af}.
On the other hand, a less studied signature of chaos which, in all known examples, agrees with the spectral definition is Local-Operator Entanglement (LOE). It is a well-justified measure of dynamical complexity and quantum chaos~\cite{Prosen2007,Prosen2007a}.

In this Letter, we give a novel interpretation for the OTOC by showing that it serves as a probe of LOE. In doing so, we will uncover simple cases where a dynamics is scrambling as signified by an exponential OTOC scaling, yet is not chaotic as demonstrated by the absence of linear LOE entropy growth. Along the way, we derive exact analytical results for a class of many-body local circuits~\cite{Akila_2016,Bertini2019exact}, including a novel exact computation of the OTOC which is of independent interest. 
Our results show that scrambling is necessary but not sufficient for quantum chaos.
}

The LOE is a measure of the complexity scaling of a Heisenberg {evolved} operator {$V_t := U_t^\dg (V \otimes \id_{\bar{B}}) U_t$} {for a local operator $V$~\cite{Prosen2007}}. {We consider an arbitrary isolated system with finite local Hilbert space dimension}. $V$ acts {locally} on a space $\mathcal{H}_B$, whereas $V_t$ {generally has support} on the full system $\mc{H}_S = \mc{H}_B \otimes \mc{H}_{\bar{B}}$. Above, {$U_t$ is an arbitrary time evolution operator.}

Specifically, the LOE is the entanglement of the Choi state of an initially local, unitary and traceless Heisenberg operator $V_t$,
\begin{equation} \label{eq:LOEdef}
    \ket{V_t} := ( V_t \otimes \id) \ket{\phi^+}, 
\end{equation}
where $\ket{\phi^+}$ is the maximally entangled state over a doubled space. As this is a pure quantum state, we can analyze its static quantum mechanical properties such as its entanglement (LOE). This entanglement can be computed across any bipartition and for any appropriate metric, such as $k-$R\'enyi entanglement entropy.

Despite not being as popular as notions of chaos based on Hamiltonian or Floquet spectra~\cite{BerryTabor1977,bgs1984}, LOE is an attractive candidate for a dynamical signature of quantum chaos in the context many-body systems. In particular: (i) It classifies the hardness of simulating the operator Heisenberg dynamics with tensor networks~\cite{Prosen2007a}, (ii) a wide range of studies into physical models support that volume-law LOE signifies non-integrability, with it scaling at most logarithmically with time for (interacting) integrable systems~\cite{Prosen2007a, Prosen2009,Dubail_2017,Jonay2018,Alba2019,Kos2020,Alba2021}, and (iii) it can be understood as a sensitivity to perturbation, analogous to the classical butterfly effect~\cite{Dowling2022}. Note that the entanglement of \emph{states} in a quantum many-body system is not a signature of chaos, with even free models generally exhibiting a linear growth~\cite{Fagotti2008,Keyserlingk2018}. Further, the LOE should not be confused with the related quantity of the `operator entanglement', which is the entanglement of the full, global unitary evolution operator~\cite{Zanardi2001,Styliaris2021}. This quantity generally scales linearly with $t$ irrespective of integrability~\cite{Dubail_2017}, unless the Hamiltonian is in a localized phase~\cite{Zhou2017}.

In comparison, {OTOC scaling generally indicates operator scrambling}, and is defined as a four-point correlator with atypical time ordering~\cite{Shenker_Stanford_2014,Maldacena_Shenker_Stanford_2016,Swingle2016,Roberts2016},
\begin{equation} \label{eq:otoc}
    F(W,V_t)= \frac{1}{d}\tr[W^\dg V^\dg_t W V_t],
\end{equation}
where we compute this expectation value over the maximally mixed state $\rho_{\infty} = \id/d$. We take $V$ and $W$ to be local unitaries, which wlog are traceless; see App.~\ref{ap:traceless}. 
In this case, the OTOC quantifies how much $V_t$ and $W$ {do not commute} as a function of time, $\mathrm{Re}[F(W,V_t)]= 1- \frac{1}{2} \braket{[W,V_t]^2}$.
The appeal for OTOC stems from a semiclassical argument connecting this equation to classical Poisson brackets, which are in turn related to Lyapunov exponents of a classical process. 

Yet, it remains unclear the precise connection of the OTOC with integrability. In fact, there is controversy in when the OTOC detects chaos in a range of quantum systems without a classical analogue~\cite{pappalardi2018,Hashimoto2020,Shor_2021,Balachandran2021-af} or even with~\cite{Pilatowsky2020,Xu2020}. In this work we clarify this confusion, showing that the OTOC probes dynamical chaos, including how it can fail in this purpose, and identifying sufficient conditions when scrambling is equivalent to chaos.

\textit{OTOC in Terms of Local Operator Choi State.---} We take $W$ in Eq.~\eqref{eq:otoc} to be on acting on some potentially large subspace $\mc{H}_A $ and $V$ on a local space $\mc{H}_B$, with complement spaces defined such that the whole isolated system is $\mc{H}_S = \mc{H}_A \otimes \mc{H}_{\bar{A}} = \mc{H}_B \otimes \mc{H}_{\bar{B}}$. These spaces are most clearly expressed via a graphical representation of ${F}(W,V_t)$:
\begin{equation}
    \frac{1}{d}\tr\left[\includegraphics[width=0.41\textwidth, valign=c]{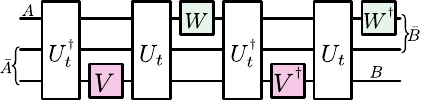}\right]. \nn
\end{equation}
We will use a bracket-prime notation to indicate a doubled space, e.g. prime $\mc{H}_{A^\prime}$ represents a copy of the space $\mc{H}_A$, while bracketed primes represent a combined double space, $\mc{H}_{A^{(\prime)}}:= \mc{H}_A \otimes \mc{H}_{A^\prime}$. For clarity, we rewrite the definition of the Choi state $\ket{V_t}$ (Eq.~\eqref{eq:LOEdef}),
\begin{equation}
    \ket{V_t}:= \frac{1}{d}\, \,\includegraphics[width=0.15\textwidth, valign=c]{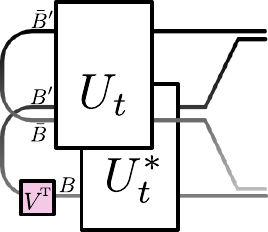} = \includegraphics[width=0.07\textwidth, valign=c]{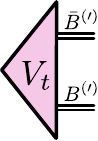}.
\end{equation}

In this setup, $V_t$ can be interpreted as the operator we are interested in, and $W$ as the \emph{probe} to the entanglement of the Choi state of this operator. The first hint of this relation can be seen by rewriting the OTOC in the following.
\begin{restatable}{obs}{otocObs} \label{obs:otocObs}
    The OTOC can be expressed in terms of the expectation value of a local unitary, $\mc{W}:=W \otimes W^*$, with respect to the Choi state of a time evolved Heisenberg operator, $\ket{V_t}$, 
    \begin{equation}
        F(W,V_t)=\bra{V_t} (\id_{\bar{A}^{(\prime)}} \ot \mc{W} ) \ket{V_t}. \label{eq:intuitive}
    \end{equation}
\end{restatable}
A proof for this can be found in App.~\ref{ap:otocObs}.

Examining the relation Eq.~\eqref{eq:intuitive}, if $\ket{V_t}$ is maximally entangled in the splitting $A^{(\prime)}:\bar{A}^{(\prime)}$, then the OTOC is equal to zero. Recalling that any maximally entangled state $\ket{\psi}$ corresponds to the Choi state of a unitary matrix $\mc{U}_\psi$, one can prove this from Eq.~\eqref{eq:intuitive} using standard graphical notation:
\begin{equation}
    \begin{split}
           F(W,V_t)&= \frac{1}{d}\, \,\includegraphics[width=0.17\textwidth, valign=c]{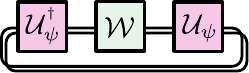} =|\tr[W] |^2 =0.  \label{eq:max_ent}
    \end{split}
\end{equation}
For example, a global Haar random time evolution, $U_t \in \mathbb{H}$, with a large total dimension $d$, will approximately give a maximally entangled $\ket{V_t}$ on average, which we show in App.~\ref{ap:max_ent}.

Alternatively, if $\ket{V_t}$ is not maximally entangled, $F(W,V_t)$ will generally be non-zero, suggesting that this quantity probes the LOE of $V_t$, through $W$. We will now quantify this relationship, showing the necessity of OTOC decay for chaotic LOE behavior.

\textit{Chaos Implies Scrambling.---} To investigate the general behavior of the OTOC, we will now compute its average value when a random traceless operator $W$ is sampled, and show that this sampling is typical. To uniformly sample a random matrix with the traceless property, a natural choice is choosing any traceless unitary $W$ and then applying a Haar random unitary channel to it, $W_{R}=R^\dg W R$ where $R \in \mathbb{H}$, the Haar measure. 
We define the averaged OTOC with respect to this traceless probe,
\begin{equation} \label{eq:haar}
    G(V_t):= \frac{1}{d}\int_{\mathbb{H}} dR \tr[ W_{R}^\dg  V^\dg_t W_{R} V_t ].
\end{equation}
Let us stress that we are \emph{not} averaging over the dynamics, and allow the time-evolution operator $U_t$ to be completely arbitrary. Hinting at a relation between the OTOC and LOE, our results will be framed in terms of $\nu_A(t)$, the (normalized) reduced density matrix of the Choi state $\ket{V_t}$ on (the doubled space) $\mc{H}_{A^{(\prime)}}=\mc{H}_A \otimes  \mc{H}_{A^\prime}$,
    \begin{equation}
        \nu_A(t):=\tr_{\bar{A}^{(\prime)}}[ \ket{V_t}\bra{V_t}]. \label{eq:nu_def}
    \end{equation}
We can then use standard techniques adapted from the Weingarten Calculus to arrive at our first main result.
\begin{restatable}{thm}{LOEOTOC} \label{thm:LOE_OTOC}
    The averaged OTOC over Haar random, traceless unitaries $W_R$ (as in Eq.~\eqref{eq:haar}) is equal to
    \begin{equation} \label{eq:mainThm}
        G(V_t)= \frac{1}{d_A^2-1} \left(d_A^2 \bra{\phi^+} \nu_A(t) \ket{\phi^+} - 1 \right),
    \end{equation}
    where $\ket{\phi^+}$ is the maximally entangled state across the doubled space $\mc{H}_{A^{(\prime)}}$. Further, this average is typical, with a single shot $F(W_R,V_t)$ exponentially likely in $(d_A \epsilon/8)^2$ to be $\epsilon-$close to $G(V_t)$.
\end{restatable}
The proof for this, together with the formal statement of the typicality result, can be found in App.~\ref{ap:1Proof}.
This result presents both an explicit expression for the average OTOC, and that a random OTOC $F(W_R,V_t)$ rarely varies from the average. This is important as the OTOC average $G(V_t)$ features in the rest of this work, and we can be assured that the average case is representative of the typical one. 

Notice that the first term in Eq.~\eqref{eq:mainThm} is proportional to the fidelity between $\nu_A(t)$ and the identity matrix Choi state $\ket{\phi^+}\!\bra{\phi^+}$. In words, this theorem states that the average OTOC is \emph{proportional to the distance between the actual reduced state of $\ket{V_t}$ on $\mc{H}_{A^{(\prime)}}$ and the state of the identity channel}. Interestingly, considering $V_t$ as a unitary channel, this fidelity is exactly equal to the entanglement fidelity of the reduced channel on $\mc{H}_A$, which in turn is proportional to the (efficiently computable) average gate fidelity~\cite{NIELSEN2002}.

\begin{figure}[t]
    \centering
    \includegraphics[width=0.37\textwidth]{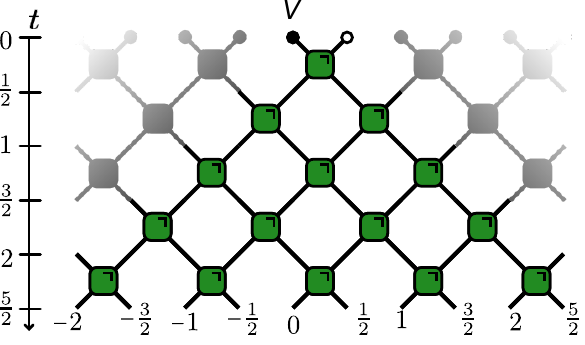}
    \caption{Brickwork circuit models of dynamics consist of repeated 2-site unitaries on a one dimensional lattice. This is given in the folded (superoperator) representation, where each `brick' represents two copies of $U$ [Eq.~\eqref{eq:u}], open (colored) circles represent the vectorized identity $\ket{\phi^+}$ ($\ket{V_t}$) state, and time goes from top to bottom. The lightcone for a single-site operator $V$ is shown in green. 
    }
    \label{fig:brick}
\end{figure}

Ref.~\cite{Styliaris2021} reports similar results to Thm~\ref{thm:LOE_OTOC}. However, there the Haar average is taken for the `bipartite OTOC' over \textit{both} $V$ and $W$, whose support bipartition the whole system. Our results are distinct, and less restrictive, in allowing the operators to have arbitrary locality, averaging over only one of the unitaries, and most importantly connecting this to LOE (chaos).

 To pick apart Theorem~\ref{thm:LOE_OTOC}, consider the example of dynamics consisting of a circuit of swaps $U_t \in \mathbb{S}$. Then, if the operator $V$ is swapped onto a site within the space $\mc{H}_A$, the OTOC takes a minimum value, 
 \begin{align}
        G(V_t)|_{U_t \in \mathbb{S}}&=\begin{cases}
    \frac{-1}{d_A^2-1},& \text{if } V_t \in \mc{B}(\mc{H}_A)\\
    1,              & \text{otherwise},
\end{cases} \label{eq:pathogolocal} 
\end{align} 
given that $V$ is taken to be traceless. Similarly, one would (approximately) get this result for $\nu_A(t) $ being (close to) any pure state which is orthogonal to $\ket{\phi^+}\!\bra{\phi^+}$. This is an example of scrambling without chaos: a minimal OTOC is achieved for an integrable dynamics. In fact, it is a simple example of a wide class of local circuit models, which we will analyze later in Results~\ref{thm:scrambNoChaos}-\ref{thm:completeChaos}. Alternatively, we saw earlier in Eq.~\eqref{eq:max_ent} that a maximally entangled $\ket{V_t}$ leads to a small OTOC. This begs the question of what the OTOC tells us if $\ket{V_t}$ is partially entangled? Can we further understand this OTOC average as a quantitative probe to the time-scaling of the LOE? 

We answer this with the following two bounds, in terms of two different entanglement measures.
\begin{restatable}{thm}{GeoEnt} \label{thm:GeoEnt}
    \textbf{(Scrambling is Necessary for Chaos)} The OTOC, averaged over traceless unitary operators $W \in \mc{B}(\mc{H}_A)$, is bounded by the entanglement on $\mc{H}_{A^{(\prime)}}$ of the time-evolved local operator $V_t$:
    \begin{enumerate} [label=\textbf{\Alph*}.]
        \item For geometric measure of entanglement, $E_G(\ket{\phi}) := 1 - \underset{\ket{\psi_{A^{(\prime)}}\psi_{\bar{A}^{(\prime)}}}}{\mathrm{max}} |\braket{\psi_{A^{(\prime)}}\psi_{\bar{A}^{(\prime)}}|\phi}|^2$, where the maximum is over all product states $\ket{\psi_{A^{(\prime)}}\psi_{\bar{A}^{(\prime)}}}$, $G(V_t)$ satisfies
        \begin{equation} \label{eq:geoEntBound}
            G(V_t) \leq 1 - \frac{d_A^2}{d_A^2-1}E_G(\ket{V_t}).
        \end{equation} 
    \item  For the 2-{R\'enyi} entropy $S^{(2)}(\nu) := -\log(\tr[\nu^2])$, $G(V_t)$ satisfies
    \begin{equation} \label{eq:ReyniBound}
        G(V_t)\leq \frac{1}{d_A^2-1} \left(d_A^2 \ex^{-\frac{1}{2} S^{(2)}(\nu_A(t)) } - 1 \right).
    \end{equation}
    \end{enumerate}
\end{restatable}
A proof for this builds on Thm.~\ref{thm:LOE_OTOC}, and is supplied in App.~\ref{ap:2Proof}. Note that we only used the inequality $\mc{D}(\nu_A(t),\ket{\phi^+}\!\bra{\phi^+}) \leq {\max}_{\ket{\psi}}(\mc{D}(\nu_A(t),\ket{\psi}\!\bra{\psi}))$ for some distance metric $\mc{D}$, to arrive at Eq.~\eqref{eq:geoEntBound}. Therefore it is likely relatively tight for a generic evolution, where $V_t$ does not recohere into a local, pure unitary channel. Indeed, from numerics we notice that Eq.~\eqref{eq:geoEntBound} seems to be tighter than Eq.~\eqref{eq:ReyniBound}. However, in general geometric measures are not practically accessible due to the required optimization over all separable states. On the other hand, the {R\'enyi} 2-entropy is. {We note that a similar result is derived in Ref.~\cite{swingle2020}, where they lower bound LOE by an extensive sum of OTOCs, in contrast to our bound in terms of a single, average OTOC.} The bound \eqref{eq:ReyniBound} is our main result, and will be investigated in the remainder of this work.

The LOE entropy grows at fastest linearly, with strong evidence that this maximal scaling is saturated iff the dynamics are chaotic~\cite{Jonay2018,Kos2020}, compared to logarithmic growth for integrable systems~\cite{Prosen2007a,Prosen2009,Muth2011,Dubail_2017,Alba2019,Alba2021}. Therefore, Eq.~\eqref{eq:ReyniBound} gives us a bound on scrambling
\begin{equation}
    \begin{split}
     G(V_t)\lesssim&\begin{cases}
   B \exp(-\alpha t) ,& \text{if $U_t$ chaotic, } \\
    C t^{-\alpha},              & \text{if $U_t$ regular}.
    \end{cases} \label{eq:scaling}
    \end{split}
\end{equation}
This should not be confused with the lower bounds on thermally regulated OTOCs from Refs.~\cite{Maldacena_Shenker_Stanford_2016,Srednicki2019otoc}, valid for fast-scrambling systems.

Theorem~\ref{thm:GeoEnt} therefore states that fast decay of the OTOC is necessary for chaos. However, the counter argument is not necessarily true, i.e. the bounds in Eq.~\eqref{eq:scaling} are not necessarily tight. We will now examine our results through classes of local circuit models, to uncover: (i) When the OTOC decays fast for slowly decaying LOE; i.e. scrambling without chaos, and (ii) When Eq.~\eqref{eq:ReyniBound} is saturated.

\textit{Application to Local Circuit Models.---} `Brickwork' circuits consist of layers of two-body unitary gates which are applied to next-neighbor sites on a lattice (see Fig.~\ref{fig:brick})~\cite{Nahum2017,Nahum2018,Keyserlingk2018,amos2018,Khemani2018,Fisher2023}, 
\begin{equation}
    U:\mc{H}_{\mathrm{i}_1}\otimes \mc{H}_{\mathrm{i}_2} \to \mc{H}_{\mathrm{o}_1}\otimes \mc{H}_{\mathrm{o}_2} . \label{eq:u}
\end{equation}
It is necessary to introduce some notation. We take the initial operator $V$ to have support on a single site which we specify wlog to be at $y=0$, where both the sites and time steps are labeled with half integers as in Fig.~\ref{fig:brick}. Then the disjoint spaces $\mc{H}_{A}$ and $\mc{H}_{\bar{A}}$ are labeled by the list of integers of the spins they contain, $\ell_A$ and $\ell_{\bar{A}}$ respectively. The following results will cover two exclusive cases: {when $\mc{H}_A$ is connected and contains either the left light cone edge ($t\in \ell_A$), or the right edge  ($t \in \ell_{\bar{A}}$)}. {This is a technical restraint, related to what is analytically computable for so-called dual unitary dynamics, which we later describe~\cite{Bertini2019exact}.} Finally, we define {convenient coordinates} $x_{\pm}:=t\pm a$, {with $a$ the edge of region $\mc{H}_{A^{(\prime)}}$ within the lightcone, such that:}
\begin{equation}
     \ket{V_t}=\includegraphics[width=0.4\textwidth, valign=c]{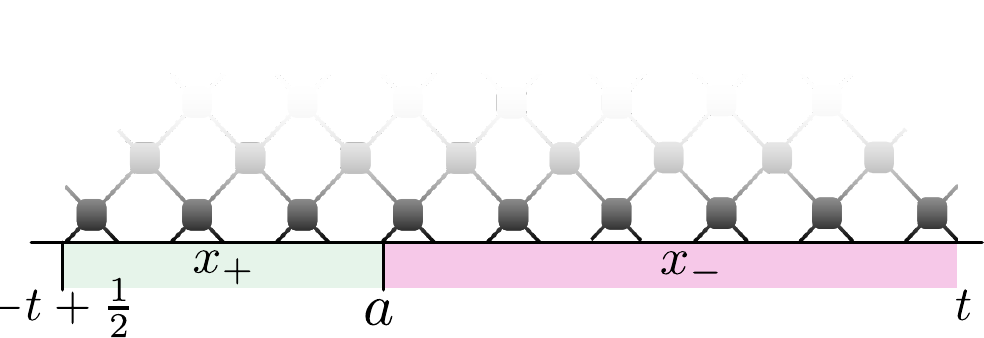} \nn
\end{equation}

In completely general brickwork unitary circuits, both the OTOC and the LOE entropy {scale} in terms of the same matrix $G(V_t) \sim (\mc{T}_{x_+}[U])^{x_-} \sim e^{-\frac{1}{2} S^{(2)}(\nu_A(t))}$. {Informally, $\mc{T}_{x_+}[U]$ is a spacetime transfer matrix, representing the contraction of the dual unitary circuit along the lightcone direction. This is a well studied object~\cite{Kos2020,Claeys2020,DeepTherm2022}, and we define it in detail in App.~\ref{ap:Pavel_TransferMatrix}.}
This leads to our first result on local circuit OTOC behavior.\textsuperscript{\footnote{Note that the following observation is not valid in some rare cases, where there could be a zero prefactor in front of the leading order on only one side of Eq.~\eqref{eq:mainThm}.}}
\begin{restatable}{obs}{brickObs}\label{obs:brickwork}
    When the term $-\frac{1}{d_A^2-1}$ from Eq.~\eqref{eq:mainThm} can be neglected, both sides of the inequality Eq.~\eqref{eq:ReyniBound} have generically the same leading order scaling for large $x_-$, but constant $x_+$. 
\end{restatable} 
A proof and detailed explanation of this result can be found in App.~\ref{ap:Pavel_TransferMatrix}. There we also further support this result with numerical examples of Haar random unitary bricks, which show similar scaling for both sides of Eq.~\eqref{eq:ReyniBound}. 
In contrast, we will now delve into a specific model where the scaling does not match, in order to highlight that scrambling is distinct from chaos.

Consider the Floquet interacting XXZ model on qubits, consisting of a brickwork dynamics (see Fig.~\ref{fig:brick}) with two-site unitary 
\begin{equation}
    U_{\mathrm{XXZ}}\!=\!\exp\left[\!-i\left( \frac{\pi}{4}\sigma_x \!\otimes \!\sigma_x \!+\! \frac{\pi}{4}\sigma_y \!\otimes\! \sigma_y \!+\! J \sigma_z \!\otimes \!\sigma_z \right) \!\right], 
    \label{eq:XXZ}
\end{equation} 
where $J$ is a free parameter. We have set the parameter in front of $\sigma_x \otimes \sigma_x$ and $\sigma_y \otimes \sigma_y$ to $\pi/4$ to impose \emph{dual-unitarity}~\cite{Vanicat2018,Ljubotina2019,Bertini2019exact}, which we later discuss. Moreover, we additionally specify that $J\neq \pi/4$, as $J=\pi/4$ yields the SWAP circuit as in Eq.~\eqref{eq:pathogolocal}. This dynamics is not chaotic. In particular, the LOE scales logarithmically with time, characteristic of interacting integrable
models~\cite{Muth2011,Alba2021}. However, we will see that the OTOC decays exponentially for all times (or is constantly minimal), indicative of strong scrambling. 
\begin{restatable}{thm}{scrambNoChaos} \label{thm:scrambNoChaos}
    \textbf{(Scrambling without Chaos)} The Floquet dual-unitary XXZ model~\eqref{eq:XXZ} produces an exponentially decaying OTOC. Concretely, for a single site operator $V$, Eq.~\eqref{eq:G_du} reduces to
     \begin{align}
        G(V_t)|_{\mathrm{XXZ}}&=\begin{cases}
    \frac{-1}{d_A^2-1},& \text{if } t \in \ell_A \\
    \beta \ex^{- \alpha (x_-) } + (1-\beta),              & \text{if } t \in \ell_{\bar{A}}.
    \end{cases} \label{eq:xxz_result1}
    \end{align} 
   with positive constants $\alpha$ and $\beta$ reported in Eq.~\eqref{eq:AlphaBeta}. 
    For any $V$ orthogonal to $\sigma_z$, the constants are such that $G(V_t)$ decays to a minimal (negative) value.
\end{restatable}
A proof for this relies on the Thms.~\ref{thm:LOE_OTOC} and \ref{thm:du}, and is supplied in App.~\ref{ap:Pavel_TransferMatrix}.
 The fact that the OTOC exhibits (maximal) exponential decay for this clearly integrable model is stark evidence of the distinction between scrambling and chaos. This lays bare the main thesis of this Letter: while the OTOC will always bear witness to chaos, there exists a wide variety of dynamics that are scrambling but not chaotic. In other words, \emph{the decay of the OTOC is necessary but not sufficient for chaos}.

Eq.~\eqref{eq:XXZ} is a particular example of a dual-unitary model~\cite{Bertini2019exact} (denoted by $\mathbb{D}$), which spread information with maximally velocity~\cite{Bertini2019KI,Claeys2020,Beretini2020otoc,Harrow2022}. Using Eq.~\eqref{eq:mainThm}, we can actually compute the average OTOC for this entire class of dynamics. In these local circuits each brick is unitary in both the space and time direction, which enables analytic computations of a range of quantities such as arbitrary local two-time correlation functions~\cite{Bertini2019exact}. Far beyond the trivial swap circuit \eqref{eq:pathogolocal}, these models are generically chaotic~\cite{Bertini2021random} and include for example the (chaotic) self-dual kicked Ising model~\cite{Akila_2016,Bertini2018KI,Bertini2019KI} and the (integrable) Floquet Heisenberg XXZ model~\cite{Vanicat2018,Ljubotina2019}, as in Eq.~\eqref{eq:XXZ}.

In terms of the (doubled) local, bipartite Hilbert spaces as in Eq.~\eqref{eq:u}, we define the CPTP maps $\mc{M}_+:= \bra{\id}_{\mathrm{i}_1^{(\prime)}} U^* \otimes U \ket{\id}_{\mathrm{o}_2^{(\prime)}}$ and $\mc{M}_-:= \bra{\id}_{\mathrm{i}_2^{(\prime)}} U^* \otimes U \ket{\id}_{\mathrm{o}_1^{(\prime)}}$. These local maps govern the decay of $2-$point correlations in $\mathbb{D}$~\cite{Bertini2019exact}. 
We can exactly express the OTOC average in terms of these maps.
\begin{restatable}{thm}{du}\label{thm:du} \textbf{(OTOC in Dual Unitary Circuits)}
    For evolution according to dual unitary circuits $\mathbb{D}$, the average OTOC is exactly 
     \begin{align}
        G(V_t)|_{\mathbb{D}}&=\begin{cases}
    -\frac{1}{d_A^2-1},&  t \in \ell_A \\
    \frac{1}{d_A^2-1}({d_A^2 \bra{V} \mc{M}_+^{x_+} \mc{M}_-^{x_+} \ket{V} - 1}),              & t \in \ell_{\bar{A}}.
    \end{cases} \label{eq:G_du} 
    \end{align}
\end{restatable}
A proof for this is supplied in App.~\ref{ap:Pavel_TransferMatrix}.
We stress that our result for $G(t)$ in Eq.~\eqref{eq:G_du} relies \emph{only} on the dual unitarity property, both in the chaotic and non-chaotic cases. This is in contrast to previous work computing OTOCs in $\mathbb{D}$, which require the `completely chaotic' assumption~\cite{Claeys2020,Beretini2020otoc}. Theorem~\ref{thm:du} is valid for arbitrary $t$ without additional averaging or conditions, and $W$ may have arbitrarily large support in contrast to the exclusively single site operators considered in Refs.~\cite{Claeys2020,Beretini2020otoc}.

We can further specify the dual unitary dynamics to be `completely chaotic'~\cite{Kos2020} (`maximally chaotic' in Ref.~\cite{Claeys2020}), defined by the property that the eigenvectors with eigenvalue one of the transfer matrix $\mc{T}$ discussed around Observation~\ref{obs:brickwork} are limited to a minimal set; see App.~\ref{ap:Pavel_TransferMatrix}.
This property generically holds, but is violated if there are additional symmetries (e.g. kicked Ising model) or additional local conservation laws (e.g. trotterized XXZ model). This property leads to a precise equivalence between the LOE and OTOC. 
\begin{restatable}{crllr}{completelyChaotic} \label{thm:completeChaos}
     For $x_-$ kept fixed, Eq.~\eqref{eq:G_du} can be expressed using the known inequalities of {R\'enyi}-2 LOE entropy,
     \begin{align} \label{eq:asymp_du}
         G(V_t)|_{U_t \in \mathbb{D}} &\leq \exp\left[ \underset{x_- \to \infty}{\lim} \frac{-1}{2}S^{(2)}(\nu_A(t) ) \right],
     \end{align}
     where equality holds exactly for completely chaotic dual unitary circuits for $|\lambda| \geq d^{-1/2}$ and $x_+$ large. Here, $|\lambda|$ is the largest non-trivial eigenvalue of $\mathcal{M}_-$.
\end{restatable}
A proof for this is supplied in App.~\ref{ap:Pavel_TransferMatrix}, and uses results from Ref.~\cite{Kos2020}.
Notice here that Eq.~\eqref{eq:asymp_du} is exactly equivalent to Eq.~\eqref{eq:ReyniBound} for $d_A \gg 1$.  That is, asymptotically the average OTOC is proportional to LOE in completely chaotic dual unitary circuits with $|\lambda | \geq d^{-1/2}$. This provides an important new insight into the completely chaotic assumption: for dual unitary circuits it is equivalent to
demanding equality in Eq.~\eqref{eq:ReyniBound}.

\textit{Conclusions and Discussion.---} 
In this Letter, we have demonstrated that the Out-of-Time-Ordered Correlator (OTOC) probes of the Local Operator Entanglement (LOE) of the time-evolving operator $V_t$ (Results~\ref{obs:otocObs}-~\ref{thm:GeoEnt}). This means that formally, scrambling is strictly necessary for chaos. To explore this relationship, we examined OTOCs in dual-unitary circuits (Theorem~\ref{thm:du}), including an explicit example of an integrable dynamics where the OTOC exponentially decays for all times; maximal scrambling without chaos (Theorem~\ref{thm:scrambNoChaos}). Finally, we also determined generic dual-unitary conditions that defines when LOE and OTOC scaling coincides -- including requiring the completely chaotic property (Corollary~\ref{thm:completeChaos}). It would be interesting to extend this and determine the class of models which saturate the bound \eqref{eq:ReyniBound}. We suspect that the members of this class share other interesting properties.

In our results, the (ultra-)local unitary operator $V$ was left unspecified, but its exact choice may influence computations (cf. Thm~\ref{thm:scrambNoChaos}).
Often, $V$ is averaged over~\cite{Roberts2017-en,Yan2020-ae}, but one can take a more subtle approach and define a density operator that encodes all possible OTOCs or local Heisenberg operators~\cite{Zonnios2022,Dowling2022}. We extend our main results to this operator-free setting in App.~\ref{ap:OTOT}.

Finally, recently there has been interest in higher point OTOC generalizations, which are thought to probe finer structure of randomness~\cite{Roberts2017-en,Oliviero2021-rf,Leone2021-ov,Leone2021quantumchaosis}. We can extend present results to connect these to a novel generalization of LOE, but leave this to a future work.

\begin{acknowledgments}
    ND thanks G. A. L. White for useful technical discussions and PK thanks G. Styliaris for fruitful discussions. We acknowledge funding from the DAAD Australia-Germany Joint Research Cooperation Scheme through the project 57445566. ND is supported by an Australian Government Research Training Program Scholarship and the Monash Graduate Excellence Scholarship. PK acknowledges financial support from the Alexander von Humboldt Foundation. KM acknowledges support from the Australian Research Council Future Fellowship FT160100073, Discovery Projects DP210100597 and DP220101793, and the International Quantum U Tech Accelerator award by the US Air Force Research Laboratory. 
\end{acknowledgments}

%
\newpage

\onecolumngrid
\appendix

\tableofcontents

\section{Traceless Operators for the OTOC} \label{ap:traceless}
Here we justify the choice of the operators $V$ and $W$ in the OTOC to be traceless. 

Any operator $W$ can be written as sum of a traceless $W^\prime$ and a constant part proportional to the identity,
    $W =  W^\prime + \frac{\id}{d_B} \tr[W]. $
Using this, the OTOC \eqref{eq:otoc} reduces to
\begin{align}
    F(W,V_t)=&\tr[W^\dg V_t^\dg W V_t] \nn\\
    =& \frac{|\tr[W]|^2}{d_A^2} \tr[\id V_t^\dg \id V_t] + \frac{\tr[W^\dg]}{d_A} \tr[\id V_t^\dg W^\prime V_t]+ \frac{\tr[W]}{d_A} \tr[(W^\prime)^\dg V_t^\dg \id  V_t] +\tr[(W^\prime)^\dg V_t^\dg W^\prime V_t] \\
    =&\tr[(W^\prime)^\dg V_t^\dg W^\prime V_t] +  \mathrm{const}, \nn
\end{align}
where we have used the unitary property of $V_t$ and that $\tr[W^\prime] = 0$. Therefore the only non-trivial part of $W$ that leads to OTOC scaling is the traceless unitary $W^\prime$, and this fact similarly holds for $V$. Hence wlog we assume $W$ and $V$ to be traceless throughout this work.

\section{Proof that OTOC is an Expectation Value of a Unitary} \label{ap:otocObs}
\otocObs*

\begin{proof}
    Rewriting the OTOC over a doubled space,
\begin{align}
            F(W,V_t)&=\tr[W^\dg V^\dg_t W V_t ] \nn \\
            &= \bra{\phi^+} (W^\dg V^\dg_t   W V_t) \otimes \id \ket{\phi^+}  \\
     &= \bra{\phi^+} (V^\dg_t W  V_t)_{A \bar{A}} \otimes (\id_{\bar{A}^\prime}\ot W^*_{A^\prime} ) \ket{\phi^+} \nn 
\end{align}
where $\ket{\phi^+}$ is the Choi state of the identity map across the total doubled system, $\mc{H}_{S^{(\prime)}}=\mc{H}_S \otimes \mc{H}_{S^\prime}=(\mc{H}_A \otimes \mc{H}_{\bar{A}})  \otimes \mc{H}_{A^\prime} \otimes \mc{H}_{\bar{A}^\prime}$, and we have used the identity $(AB \otimes \id )\ket{\phi^+} = (A \otimes B^T )\ket{\phi^+}$. We have also introduced subscripts to make clear the different Hilbert spaces used here. Then, recalling the definition $\ket{V_t}:=( V_t \otimes \id)\ket{\phi^+} $ (Eq.~\eqref{eq:LOEdef}), we have that 
\begin{align}
     F(W,V_t)&= \bra{V_t} (\id_{\bar{A}} \ot W_A ) \otimes (\id_{\bar{A}} \ot W^*_{A^\prime} ) \ket{V_t}  \nn\\
    &= \bra{V_t} (\id_{\bar{A}^{(\prime)}} \ot \mc{W} ) \ket{V_t} \label{eq:intuitiveproof}      
\end{align}
where $\mc{W}:=W \otimes W^*$ is a superoperator acting over both copies of the system, on the vector $\ket{V_t} \in \mc{H}_{S^{(\prime)}}$. 

We also reproduce the above proof graphically, where we believe it is far easier to follow:
\begin{align}
    F(W,V_t)&= \frac{1}{d} \tr\left[\includegraphics[width=0.41\textwidth, valign=c]{otoc_def.pdf}\right] \nn \\
    &=\frac{1}{d} \, \,\includegraphics[scale=1.7, valign=c]{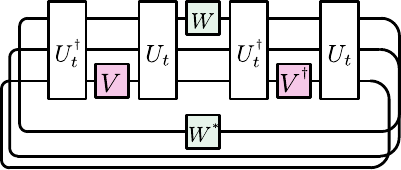} \\
    &=\frac{1}{d}\, \, \includegraphics[scale=1.7, valign=c]{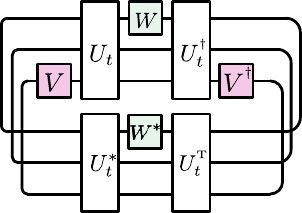} \nn \\
    &= \frac{1}{d}\, \,\includegraphics[width=0.2\textwidth, valign=c]{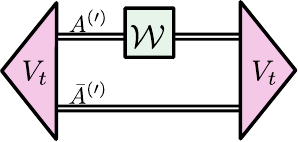}. \nn 
\end{align}
\end{proof}

\section{Proof that maximum entangled Heisenberg operators leads to small OTOC} \label{ap:max_ent}
To derive Eq.~\eqref{eq:max_ent} algebraically, we start directly from Eq.~\eqref{eq:intuitive} and use that the Schmidt spectrum of a maximally entangled state is uniform, 
\begin{align} \label{eq:schmidt_proof}
    F(W,V_t)&=\sum_{i,j} \frac{1}{d_A^2} \bra{\alpha_i} \mc{W} \ket{\alpha_i} \braket{\beta_i|\beta_j} =\frac{1}{d_A^2} \sum_{i}  \bra{\alpha_i} \mc{W} \ket{\alpha_i} =\frac{1}{d_A^2}  \tr[\mc{W}]=\frac{1}{d_A^2} 
 |\tr[W]|^2 = 0. 
\end{align}
Here, the final equality is due to imposing the traceless property of $W$.

Further, choosing the global dynamics to be Haar random, on average one finds that 
\begin{align}
    \braket{\nu_A }_{U_t \in \mathbb{H}} &= \frac{1}{d^2-1}(d^2 \rho_{\infty} - \ket{\phi^+ }\bra{\phi^+} ) \approx \rho_{\infty}. 
\end{align}
where $\rho_{\infty} = \id/d $ is the infinite temperature state, and where we have assumed $d\gg 1$ in the final approximation. This can be proven directly, using the expression for the reduced state of $\ket{V_t}$ on $A^{(\prime)}$ (Eq.~\eqref{eq:nu_def}),
\begin{equation}
    \nu_A(t):=\tr_{\bar{A}^{(\prime)}}[ \ket{V_t}\bra{V_t}],
\end{equation}
and applying the expression for the $2-$fold Haar average, Eq.~\eqref{eq:2fold}.

\section{Explicit Expression for Generic OTOC} \label{ap:1Proof}
\LOEOTOC*

We will first prove the exact form of the average value $G(V_t)$, Eq.~\eqref{eq:mainThm}, and then proceed to separately prove that this average is typical. 
\begin{proof}
Choose the probing unitary $W$ to be Haar random, while preserving its traceless property. That is, consider the average quantity,
\begin{equation}
    G := \frac{1}{d}\int_{R \in \mathbb{H}} dR \tr[R^\dg W^\dg R  V_t^\dg R^\dg W R V_t ].
\end{equation}
Recall that each of $W$, $V$, and $U$ are local unitary matrices, acting on the spaces $A$, $B$ and $A$ respectively,
\begin{equation}
    W \equiv W_A \otimes \id_{\bar{A}}.
\end{equation}
Note also that we take $A$ and $B$ to be disjoint, $B \in \bar{A}$. We write ${W}_R \equiv R^\dg W R \in \mc{H}_A$. $V_t$ is in general global for large enough $t$, so $V_t \in \mc{H}_S$. Then 
\begin{align}
    G(V_t)&= \frac{1}{d}\int dR \tr[{W}_R^\dg V_t^\dg {W}_R V_t] \nn \\
    &=\frac{1}{d}\int dR \tr[  (\id_{\bar{A} } \otimes ({W}_R^\dg)_A) V_t^\dg  (\id_{\bar{A} } \otimes ({W}_R)_A) V_t] \label{eq:proofA}\\
    &= \frac{1}{d}\int dR \tr[  V_t^\dg (\id_{\bar{A} } \otimes ({W}_R^\dg)_A)   \otimes V_t  (\id_{\bar{A}^\prime } \otimes ({W}_R)_{A^\prime}) \mathbb{S}_{S S^\prime}] \nn \\
    &=\frac{1}{d}\int dR \tr[  (V_t^\dg \otimes V_t) \Big( (R^\dg)^{\otimes 2} ({W}^\dg \otimes {W}  ) (R^{\otimes 2}) \otimes \id_{\bar{A}^{(\prime)} }\Big)  \mathbb{S}] \nn, 
\end{align}
where $\mathbb{S} $ is the swap operation.
Written in this form in the doubled space, we can see that the Haar average over $U$ is a $2-$fold channel of the quantity ${W}^\dg_A \otimes {W}_{A^\prime} $. We can compute this exactly using well-known formula which can be derived from Weingarten calculus, which states that the 2-fold Haar average (over $U \in \mc{H}$, with dimension $d$) of some tensor $X \in \mc{H}\otimes \mc{H}$ is~\cite{Roberts2017-en}
\begin{align}
    &\Phi^{(2)}_\mathrm{Haar}(X):= \int dU U \otimes U (X) U^\dg \otimes U^\dg \label{eq:2fold}\\
    &\quad= \frac{1}{d^2-1} \left(\id \tr[X] + \mathbb{S} \tr[\mathbb{S} X] -\frac{1}{d}\mathbb{S} \tr[X] -\frac{1}{d} \id \tr[\mathbb{S} X] \right), \nn
\end{align}
Note also that for any $X$, by definition the $2-$fold Haar average is equal to the $2-$fold average over a unitary $2-$design, that is $\Phi^{(2)}_\mathrm{Haar}(X) = \Phi^{(2)}_\mathrm{2-design}(X)$. 

Examining the terms of the identity \eqref{eq:2fold} for $X \equiv {W}^\dg_A \otimes {W}_{A^\prime} $, note that 
\begin{align}
    &\tr[{W}^\dg_A \otimes {W}_{A^\prime}] = \tr[ {W}^\dg_A ] \tr[   {W}_{A^\prime}] = 0, \text{ and} \\
    &\tr[{W}^\dg_A \otimes {W}_{A^\prime} \mathbb{S}_{A A^\prime}] = \tr[{W}^\dg_A {W}_{A}] = \tr[\id_A] = d_A
\end{align}
where the first line is due to the choice of $W$ being traceless, and the second line comes from $W$ being unitary. Therefore in this case only the second and fourth terms in Eq.~\eqref{eq:2fold} are non-zero. Subbing this into Eq.~\eqref{eq:proofA}, we arrive at 
\begin{align}
    G(V_t)&=\frac{1}{d(d_A^2-1)} \tr[  (V_t^\dg \otimes V_t) \Big( \big( d_A \mathbb{S}_{A A^\prime } - \id_{A A^\prime}\big) \otimes \id_{\bar{A}^{(\prime)} } \Big)  \mathbb{S}_{S S^\prime}] \nn \\
    &=\frac{1}{d(d_A^2-1)} \left(d_A \tr[(V_t^\dg \otimes V_t) \id_{A^{(\prime)}} \mathbb{S}_{\bar{A}\bar{A}^\prime}] - \tr[V_t^\dg V_t ] \right) \label{eq:proofBB}\\
    &= \frac{1}{d(d_A^2-1)} \left(d_A \tr_{\bar{A}^{(\prime)}}[(\tr_A[V_t^\dg] \otimes \tr_{A^\prime}[V_t]) \mathbb{S}_{\bar{A}\bar{A}^\prime}] - d \right) \nn \\
    &=\frac{1}{d(d_A^2-1)} \left(d_A \tr_{\bar{A}}[(\tr_A[V_t^\dg] \tr_{A^\prime}[V_t])] - d \right) \nn
\end{align}
Now we notice that 
\begin{align}
    \tr_A[V_t] = d_A \braket{\phi^+|_{A^{(\prime)}} (V_t \otimes \id) |\phi^+}_{A^{(\prime)}}.
\end{align}
We have brought out the supernormalization here and in the following, such that $\ket{\phi^+}$ is a normalized quantum state. Subbing this into Eq.~\eqref{eq:proofBB}, we get 
\begin{align}
    G&=\frac{1}{d(d_A^2-1)} \left(d_A \tr_{\bar{A}}[|d_A \braket{\phi^+|_{A^{(\prime)}} (V_t \otimes \id) |\phi^+}_{A^{(\prime)}}|^2] - d \right) \label{eq:proofBBB} \\
    &=\frac{1}{d(d_A^2-1)} \left(d_A^3 d_{\bar{A}} \tr_{\bar{A}^{(\prime)}}[| \braket{\phi^+|_{A^{(\prime)}} (V_t \otimes \id) |\phi^+}_{A^{(\prime)}}|^2 \ket{\phi^+}_{\bar{A}^{(\prime)}}\bra{\phi^+}_{\bar{A}^{(\prime)}}] - d \right) \nn
\end{align}
Now we recall the definition of the Choi state of the time-evolved local operator $V_t$, in terms of the maximally mixed state $\ket{\phi^+}$ on the doubled space $\mc{H}_S \otimes \mc{H}_{S^\prime}$
\begin{equation}
    \ket{V_t} := 
    (\id_{S^\prime} \otimes V_t) \ket{\phi^+}.
\end{equation}
Then from Eq.~\eqref{eq:proofBBB}
\begin{align}
    G&= \frac{1}{d(d_A^2-1)} \left(d_A^3 d_{\bar{A}} \tr_{\bar{A}^{(\prime)}}[ \braket{\phi^+|V_t}\! \braket{V_t|\phi^+}_{A^{(\prime)}} ] - d \right) \nn \\
    &=\frac{1}{d(d_A^2-1)} \left(d_A^2 d \tr_{A^{(\prime)}}[ \ket{\phi^+} \bra{\phi^+}(\tr_{\bar{A}^{(\prime)}}[\ket{V_t}\bra{V_t}] )] - d \right) \label{eq:proofC}
\end{align}
We additionally define the reduced state of $\ket{V_t}$ on $\mc{H}_{A^{(\prime)}}$, 
\begin{equation}
    \nu_A(t) := \tr_{\bar{A}^{(\prime)}}[ \ket{V_t}\bra{V_t}].
\end{equation}
This is a normalized density matrix, and as the reduced state of a pure state it be used can measure the entanglement of this state. Then Eq.~\eqref{eq:proofC} reduces to,
\begin{align}
    G&= \frac{1}{d_A^2-1} \left(d_A^2 \tr_{A^{(\prime)}}[ \ket{\phi^+} \bra{\phi^+} \nu_A (t) ] - 1 \right)
    = \frac{1}{d_A^2-1} \left(d_A^2 \mc{F}(\nu_A(t), \ket{\phi^+}\!\bra{\phi^+}) - 1 \right)
\end{align}
Where $\mc{F}$ is the fidelity of quantum states. 
\end{proof}

We will now prove the second part of Theorem~\ref{thm:LOE_OTOC}, the typicality of the average. We first restate this as a standalone result, such that it is technically complete.

\noindent \textbf{Proposition 2.1.} \textit{The probability that the OTOC $F(W_R,V_t)$ for some Haar random, traceless unitary $W_R$ varies from the average $G(V_t)$ by more than some $\epsilon>0$, satisfies 
    \begin{equation} \label{eq:ConcentrationofMeasure}
        \mathrm{Pr}_{R \sim \mathbb{H}}\left\{ |F(W_R,V_t) -G(V_t) | \geq \epsilon \right\} \leq \exp \left(-\frac{d_A^2  \epsilon^2 }{ 64 } \right). 
    \end{equation}}
    
This means that if one chooses the traceless `probing' unitary $W$ in the OTOC $F(W,V_t)$ to be Haar random and acting on a relatively large subspace $\mc{H}_A$, then it is exponentially likely to approximately satisfy Eq.~\eqref{eq:mainThm} and all following results where $G(V_t)$ appears. Even if $\mc{H}_A$ is composed of only $\mc{O}(10^1)$ qubits, the right hand side of Eq.~\eqref{eq:ConcentrationofMeasure} gives a confidence of $\mc{O}(1/\ex^4) \sim 0.0183$ within an error of $\epsilon \sim \mc{O}(2^{-6}) \sim 0.0156$. 

\begin{proof}This proof uses some techniques from that of Proposition 3 from Ref.~\cite{Styliaris2021}, found in its Supplemental Material. However, we arrive at a smaller Lipschitz constant. $\mathbb{H}$ refers to the Haar measure, the unique, unitarily invariant measure on the space of unitary matrices $\mathbb{U}(d)$. We will be applying Levy's Lemma, a concentration of measure result.

    \noindent \textbf{Lemma 2.2.} \textit{ Levy's Lemma states that for $U$ sampled according to the Haar measure $\mathbb{H}$, $f: \mathbb{U}_d \to \mathbb{R}$ a Lipschitz continuous function with Lipschitz constant $K$, and $\epsilon >0$ then
        \begin{equation}
            \mathrm{Pr}_{U\sim \mathbb{H}} \left\{|f(U) - \braket{f(U)}_{\mathbb{H}}|\geq \epsilon \right\} \leq \exp \left(-\frac{d \epsilon^2 }{4 K^2} \right)  
        \end{equation}
        where $K$ is defined such that for all $N,M\in \mathbb{U}_d$ 
        \begin{equation}
            |f(N) - f(M) |\leq K \|N-M \|_2.
        \end{equation}}
        
    We first need to show that the function 
    \begin{equation}
        f(U)= \frac{1}{d}\tr[U^\dg W^\dg U  V_t^\dg U^\dg W U V_t ]
    \end{equation}
    is Lipschitz continuous and determine its constant $K$. We will use the shorthand notation for superoperators, for $X \in \mc{B}(\mc{H}_A)$
    \begin{equation} \label{eq:superops}
        \mc{N}(X) := N^\dg (X) N \otimes \id_{\bar{A}},
    \end{equation}
    and similarly 
    \begin{equation} \label{eq:superops1}
        \mc{M}(X) := M^\dg (X) M \otimes \id_{\bar{A}}.
    \end{equation}
    
    Then, using H\"older's inequality $|\tr[ AB]| \leq \|A  \|_{\infty} \| B\|_1$
    \begin{align}
        |f(N)-f(M)| &=  \frac{1}{d}|\tr[ V_t^\dg (\mc{N}(W^\dg)V_t \mc{N}(W) - \mc{M}(W^\dg)V_t \mc{M}(W)  ]| \nn\\
        &\leq \frac{1}{d}\|V_t \|_{\infty} \|\mc{N}(W^\dg)V_t \mc{N}(W) - \mc{M}(W^\dg)V_t \mc{M}(W) \|_1 \\
        &=\frac{1}{d} \|\mc{N}(W^\dg)V_t (\mc{N}(W) -\mc{M}(W)) - (\mc{M}(W^\dg) - \mc{N}(W^\dg))V_t \mc{M}(W) \|_1 \nn 
    \end{align}
    Here we have also used that $\| X \|_{\infty} = 1$ for unitary $X$, and added and subtracted $\mc{N}(W^\dg) V_t \mc{M}(W)$. We can now apply the triangle inequality,
    \begin{align}
        |f(N)-f(M)| &\leq \frac{1}{d} \Big( \|\mc{N}(W^\dg)V_t (\mc{N}(W) -\mc{M}(W)) \|_1 +\|(\mc{M}(W^\dg) - \mc{N}(W^\dg))V_t \mc{M}(W)  \|_1\Big)\nn \\ 
        &= \frac{1}{d} \Big( \|\mc{N}(W) -\mc{M}(W)\|_1 +\|\mc{M}(W^\dg) - \mc{N}(W^\dg) \|_1\Big)
    \end{align}
    where we have also used that Schatten p norms are unitarily invariant. Then, as $\| X\|_1 \leq \sqrt{d} \| X\|_2$
    \begin{align}
        |f(N)-f(M)| \leq&\frac{\sqrt{d}}{d}\Big( \|\mc{N}(W) -\mc{M}(W)\|_2 +\|\mc{M}(W^\dg) - \mc{N}(W^\dg) \|_2\Big) \nn \\
        =&\frac{1}{\sqrt{d}} \Big( \|\mc{N}(W) -\mc{M}(W)\|_2 +\|\mc{M}(W^\dg) - \mc{N}(W^\dg) \|_2\Big) \\
        =&\frac{2}{\sqrt{d}} \Big( \|\mc{N}(W) -\mc{M}(W)\|_2\Big) \nn\\
        =&\frac{2}{\sqrt{d}} \Big( \|({N}^\dg(W)N -{M}^\dg(W)M) \otimes \id_{\bar{A}}\|_2\Big) \nn \\
        =&\frac{2\sqrt{d_{\bar{A}}}}{\sqrt{d}} \Big( \|{N}^\dg(W)N -{M}^\dg(W)M\|_2 \Big) \nn
    \end{align}
    where we have subbed in the definitions \eqref{eq:superops} and \eqref{eq:superops1}, and used that $\|X \otimes \id_{\bar{A}} \|_2 = \|X \|_2 \|\id_{\bar{A} }\|_2 = \sqrt{d_{\bar{A}}} \|X \|_2$. We will now apply a similar triangle inequality `trick' as before,
    \begin{align}
        |f(N)-f(M)| \leq&\frac{2}{\sqrt{d_A}} \Big( \|{N}^\dg(W)N -{M}^\dg(W)M\|_2 \Big) \nn\\
        =&\frac{2}{\sqrt{d_A}} \Big( \|{N}^\dg(W)(N-M) -({M}^\dg-N^\dg)(W)M\|_2 \Big)  \\
        \leq&\frac{2}{\sqrt{d_A}} \Big( \|{N}^\dg(W)(N-M)\|_2 + \| ({M}^\dg-N^\dg)(W)M\|_2 \Big) \nn \\
        =&\frac{4}{\sqrt{d_A}} \|N-M\|_2, \nn
    \end{align}
    where we have again used the unitary invariance of Schatten norms. Therefore, a Lipschitz constant for the OTOC function $f(U)$ is
    \begin{equation}
        K=\frac{4}{\sqrt{d_A}}.
    \end{equation}
    Directly applying this to Levy's Lemma completes the proof. 
    \end{proof}


\section{LOE Upper Bounds the OTOC} \label{ap:2Proof}
\GeoEnt*

\begin{proof}
\begin{enumerate} [label=\textbf{\Alph*}.]
    \item We first define the geometric measure of entanglement for pure states across the bipartition $A:\bar{A}$,
    \begin{align} \label{eq:geoEnt}
        E_G(\ket{\phi}) :=& 1 - \underset{\ket{\psi_A\psi_{\bar{A}}}}{\mathrm{max}} |\braket{\psi_A\psi_{\bar{A}}|\phi}|^2.
    \end{align}
    The maximum is taken over all states separable in the splitting $A:\bar{A}$, and $\rho_A := \tr_{\bar{A}}[\ket{\phi} \bra{\phi}]$. Then, from Eq.~\eqref{eq:mainThm}, noticing that the first term is equal to the quantum fidelity: $\braket{\phi^+ | \nu_A(t) |\phi^+} = \mc{F}(\ket{\phi^+}\!\bra{\phi^+},\nu_A(t) )$, we have that 
    \begin{align}
        G(V_t) &=\frac{1}{d_A^2-1} \left(d_A^2 \mc{F}\left(\ket{\phi^+}\!\bra{\phi^+},\nu_A(t)\right) - 1 \right)  \nn \\
        &\leq \frac{1}{d_A^2-1} \underset{\ket{\psi_A}}{\mathrm{max}} \left[ d_A^2 \mc{F}\left(\ket{\psi_A}\!\bra{\psi_A},\nu_A(t)\right) - 1 \right] \nn \\
        &= \frac{1}{d_A^2-1} \underset{\ket{\psi_A}}{\mathrm{max}} \left[ d_A^2 \underset{U}{\mathrm{max}} |\braket{\psi_A \psi_{\bar{A}}|\id_A  \otimes  U_{\bar{A}}|V_t} |^2  - 1 \right]  \\
        &= \frac{1}{d_A^2-1} \underset{\ket{\psi_A}}{\mathrm{max}} \left[ d_A^2 \underset{\ket{\psi_{\bar{A}}}}{\mathrm{max}} |\braket{\psi_A \psi_{\bar{A}}|V_t} |^2  - 1 \right] \nn \\
        &=\frac{1}{d_A^2-1} \left[ d_A^2 \left( 1 - E_G(\ket{V_t}) \right)- 1  \right] \nn \\
        &= 1 - \frac{d_A^2}{d_A^2-1}E_G(\ket{V_t}) \nn
    \end{align}
    where the third line we have used Uhlmann's Theorem~\cite{wilde_2017}.
    \item Starting again from Eq.~\eqref{eq:mainThm}, we have that
    \begin{equation}
        G(V_t)= \frac{1}{d_A^2-1} \left(d_A^2 \braket{\phi^+ | \nu_A(t) |\phi^+}  - 1 \right) 
    \end{equation}
    Then, we can use the Cauchy-Schwarz inequality for the Hilbert Schmidt inner product, $\braket{A,B}_{\mathrm{HS}} := \tr[A^\dg B ]$,
    \begin{equation}
        \braket{\phi^+ | \nu_A(t) |\phi^+} = \langle \nu_A(t), \ket{\phi^+}\!\bra{\phi^+} \rangle_{\mathrm{HS}} \leq \sqrt{\tr[\nu_A(t)^2]}\sqrt{\braket{\phi^+|\phi^+}^2} = \sqrt{\tr[\nu_A(t)^2]}
    \end{equation}
    to arrive at 
    \begin{equation}
        G(V_t)\leq \frac{1}{d_A^2-1} \left(d_A^2 \sqrt{\tr[\nu_A(t)^2]} - 1 \right) 
    \end{equation}
    Then,
    \begin{equation}
         G(V_t)\leq \frac{1}{d_A^2-1} \left(d_A^2 \ex^{-\frac{1}{2} S^{(2)}(\nu_A(t)) } - 1 \right),
    \end{equation}
    where the $2-${R\'enyi} entropy is defined as $S^{(2)}(\nu) := -\log(\tr[\nu^2])$.
\end{enumerate}
\end{proof}

\section{OTOC vs. LOE for Brickwork Circuits} \label{ap:Pavel_TransferMatrix}
Here we investigate the scaling of {R\'enyi} Entropy $S^{(2)}(\nu_t)$ and the average OTOC $G(V_t)$ in general brickwork circuits, and therefore: (i) show that Eq.~\eqref{eq:ReyniBound} is a faithful bound according to the leading order scaling of transfer matrices which we will define (ii) numerically verify that Haar random brickwork circuits show a similar scaling (iii) compute $G(V_t)$ for dual unitary circuits, and show how they compare to known results on LOE from Ref.~\cite{Kos2020} and (iv) show that Eq.~\eqref{eq:ReyniBound} is saturated asymptotically for `completely chaotic' dual unitary circuits. We will use techniques and notation that largely follow Ref.~\cite{Kos2020}. For further information on the notation used, see therein. 

We will use the superoperator (doubled/folded) representation of the components of the local unitary circuit, as in the main text represented with calligraphic script. Graphically, it is convenient to define 
\begin{equation}
    \mc{U}:=U \otimes U^* = \includegraphics[scale=1.5, valign=c]{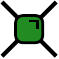} 
\end{equation}
and its transpose conjugate
\begin{equation}
    \mc{U}^*:=U^\dg \otimes U^{\mathrm{T}} = \includegraphics[scale=1.5, valign=c]{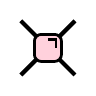} .
\end{equation}
Time runs from top to bottom in these diagrams. Note that each line represents a doubled space. Unitarity then means that 
\begin{equation}
    \tr_{\mathrm{in}}[U \otimes U^*]=\includegraphics[scale=1.5, valign=c]{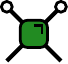}=\includegraphics[scale=1.5, valign=c]{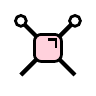}=\includegraphics[scale=1.5, valign=c]{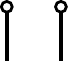},\label{eq:unitarity}
\end{equation}
and also that
\begin{equation}
    (U \otimes U^*)(U^\dg \otimes U^{\mathrm{T}})=\includegraphics[scale=1.5, valign=c]{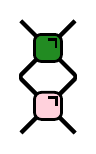}=\includegraphics[scale=1.5, valign=c]{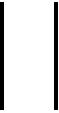}=\id.\label{eq:unitarity1}
\end{equation}
Here, white circles indicate a projection onto the maximally mixed state on the double space, in other words tracing over the space.   

Now consider a local Heisenberg operator $V$ under evolution of a brickwork unitary circuit, the folded diagrammatic representation for its Choi state is 
\begin{equation}
    \ket{V_t} = \includegraphics[scale=1.5, valign=c]{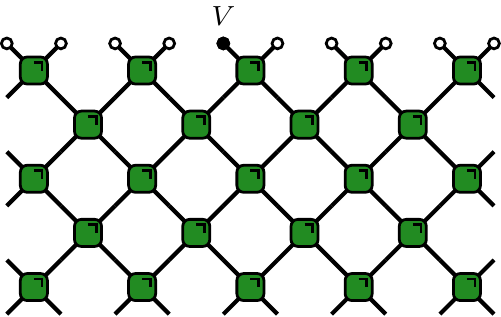} .\label{eq:brickwork}
\end{equation}
Note that every line in this diagram represents a doubled Hilbert space. Here, the black circle represents the Choi state of the local operator,
\begin{equation}
    (V \otimes \id )\ket{\phi^+}=\includegraphics[scale=1.5, valign=c]{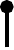} = \includegraphics[scale=1.5, valign=c]{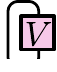}.
\end{equation}
Local circuits have a natural `lightcone' from the unitary property of the bricks: applying the graphical rule Eq.~\eqref{eq:unitarity}, Eq.~\eqref{eq:brickwork} reduces to 
\begin{equation}
    \ket{V_t} = \includegraphics[scale=1.5, valign=c]{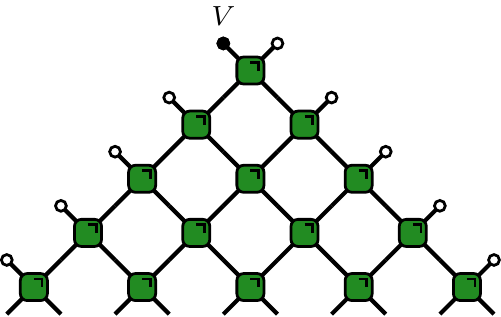} .\label{eq:brickwork_lc}
\end{equation}
Now, we can use this to come up with a graphical expression for the reduced state of $\ket{V_t}$ on $A^{(\prime)}$,
\begin{equation}
    \nu_A(t) = \tr_{\bar{A}^{(\prime)}}[ \ket{V_t} \bra{V_t}]=\includegraphics[scale=1.5, valign=c]{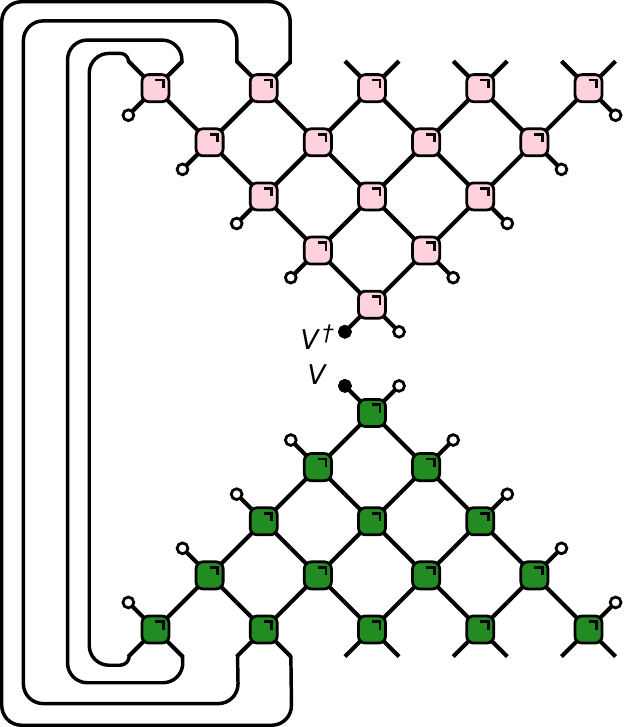}.\label{eq:nu_graph}
\end{equation}
Notice that we can further simplify the three top-left gates with the three bottom-left ones using the unitary condition \eqref{eq:unitarity1}.
The resulting expression can then be used to compute the OTOC average $G(V_t)$ (Eq.~\eqref{eq:mainThm}) and and the $2-${R\'enyi} entropy $S^{(2)}$, in order to test Eq.~\eqref{eq:ReyniBound}.

\subsection{Transfer Matrix Expressions}
\brickObs*
\begin{proof}
    We now define a transfer matrix which yields a convenient description from which to compute quantities in brickwork circuits 
\begin{equation}
    \mc{T}_s[U]=\includegraphics[scale=1.5, valign=c]{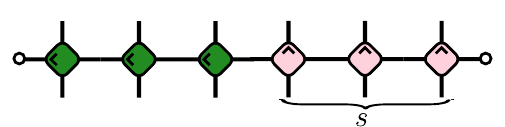}. \label{eq:tranMatrix}
\end{equation}
This transfer matrix appears in both the expression for the $2-${R\'enyi} entropy $S^{(2)}(V_t)$ and the OTOC average $G(V_t)$. In particular, for the five layer example in Eq.~\eqref{eq:nu_graph} ($t=5/2$, $a= -1/2$, $x_+= 2 $, $x_-=3$), using the unitarity graphical identities \eqref{eq:unitarity}-\eqref{eq:unitarity1} it is easy to show that 
\begin{equation}
    \braket{\phi^+|\nu_A(t)|\phi^+}= \includegraphics[scale=1.5, valign=c]{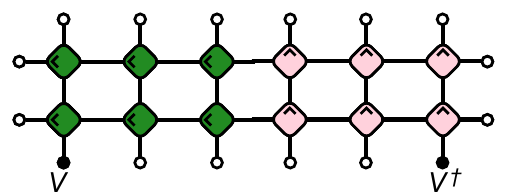} = \bra{(\phi^+)^{\otimes 6}}\mc{T}_3[U]^2 \ket{V_t \ot (\phi^+)^{\otimes 4} \ot V_t^\dg}.
\end{equation}
Whereas, for the $2-${R\'enyi} entropy similarly for the example \eqref{eq:nu_graph}, one can show that
\begin{align}
    \ex^{-S^{(2)}(\nu_A(t))}&= \includegraphics[scale=1.5, valign=c]{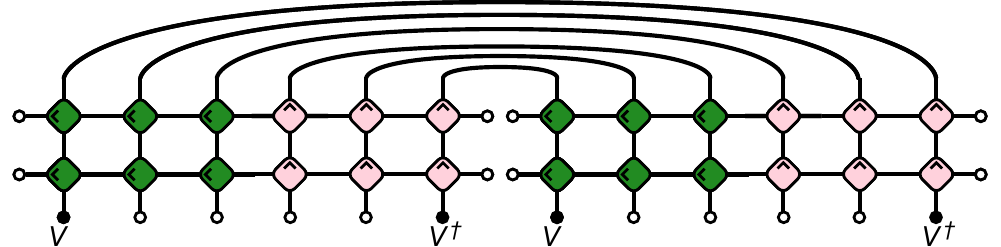} \nn\\
    &= \bra{r_6} \left(\mc{T}_3[U]^2 \ket{V_t \ot (\phi^+)^{\otimes 4} \ot V_t^\dg}\right)^{\ot 2}\label{eq:23}
\end{align}
where $\ket{r_6}$ is the `rainbow state' on six sites, i.e. the contraction of indices in the graphical representation of Eq.~\eqref{eq:23}, formally defined as~\cite{Kos2020}
\begin{equation}
    \ket{r_j}:=\frac{1}{d^j} \sum_{I_1,I_2,\dots,I_j}^{d^2} \ket{I_1 I_2 \cdots I_j } \ket{I_j I_{j-1} \cdots I_1 },\label{eq:rainbow}
\end{equation}
where $d$ refers to the local dimension of a single site here. Recalling that $x_\pm := t \pm a$, in full generality we have that 
\begin{equation}
    \begin{split}
            &\braket{\phi^+|\nu_A(t)|\phi^+} =  \bra{(\phi^+)^{\otimes 2x_-}}\mc{T}_{x_-}[U]^{{x_+}} \ket{V_t \ot (\phi^+)^{\otimes (2x_- - 2)} \ot V_t^\dg},\text{ and}\\ 
     &\ex^{-S^{(2)}(\nu_A(t))}=\bra{r_{x_-}} \left(\mc{T}_{x_-}[U]^{x_+} \ket{V_t \ot (\phi^+)^{\otimes (2x_- - 2)} \ot V_t^\dg}\right)^{\ot 2}. \label{eq:scalingTransf}
    \end{split}
\end{equation}
We note that, in fact, one should define the exponent of $\mc{T}$ as $\lceil x_+\rceil$ and the parameter of $\mc{T}$ as $\lfloor x_- \rfloor$, to account for when $x_\pm$ is half integer (this applies also to the results \ref{thm:scrambNoChaos}-\ref{thm:completeChaos}). One can circumvent this in the main text by just choosing $a$ such that $x_\pm$ is an integer. Then the leading eigenvectors of $\mc{T}_{x_-}[U]$ will dominate both the expressions in \eqref{eq:scalingTransf}, for large $x_+$, when the term $-\frac{1}{d_A^2-1}$ from Eq.~\eqref{eq:mainThm} can be neglected. Strictly speaking, $\bra{(\phi^+)^{\otimes 2x_-}}$ and/or $\bra{r_{x_-}}$ could have zero overlap with the leading eigenvector, but this generically won't happen.
Considering $\lambda$ as the leading non-trivial eigenvalue of $\mc{T}_{x_-}[U]$, this means that for a scaling $x_+$ but constant $x_-$ (i.e. time $t$ and $a$ scaling proportionally), both 
\begin{equation}
    G(V_t) \sim \lambda^{{x_+}}
\end{equation}
and 
\begin{equation}
    \sqrt{\ex^{-S^{(2)}(\nu_A(t))}} \sim \lambda^{{x_+}},
\end{equation}
Therefore, generically both sides of Eq.~\eqref{eq:ReyniBound} have the same asymptotic scaling with large $x_+$ (large $t$ but constant $d_A$).
\end{proof}

\subsection{Random Brickwork Scaling (Numerics)}
As a simple test case of Eq.~\eqref{eq:ReyniBound}, we can choose each brick of a local circuit to be chosen according to the Haar distribution (see details around Theorem~\ref{thm:LOE_OTOC}). In this case we see a similar trend for the left and right hand side of the inequality~\eqref{eq:ReyniBound}. This is presented in Fig.~\ref{fig:bound}.

\begin{figure}[h]
    \centering
    \includegraphics[width=0.45\textwidth]{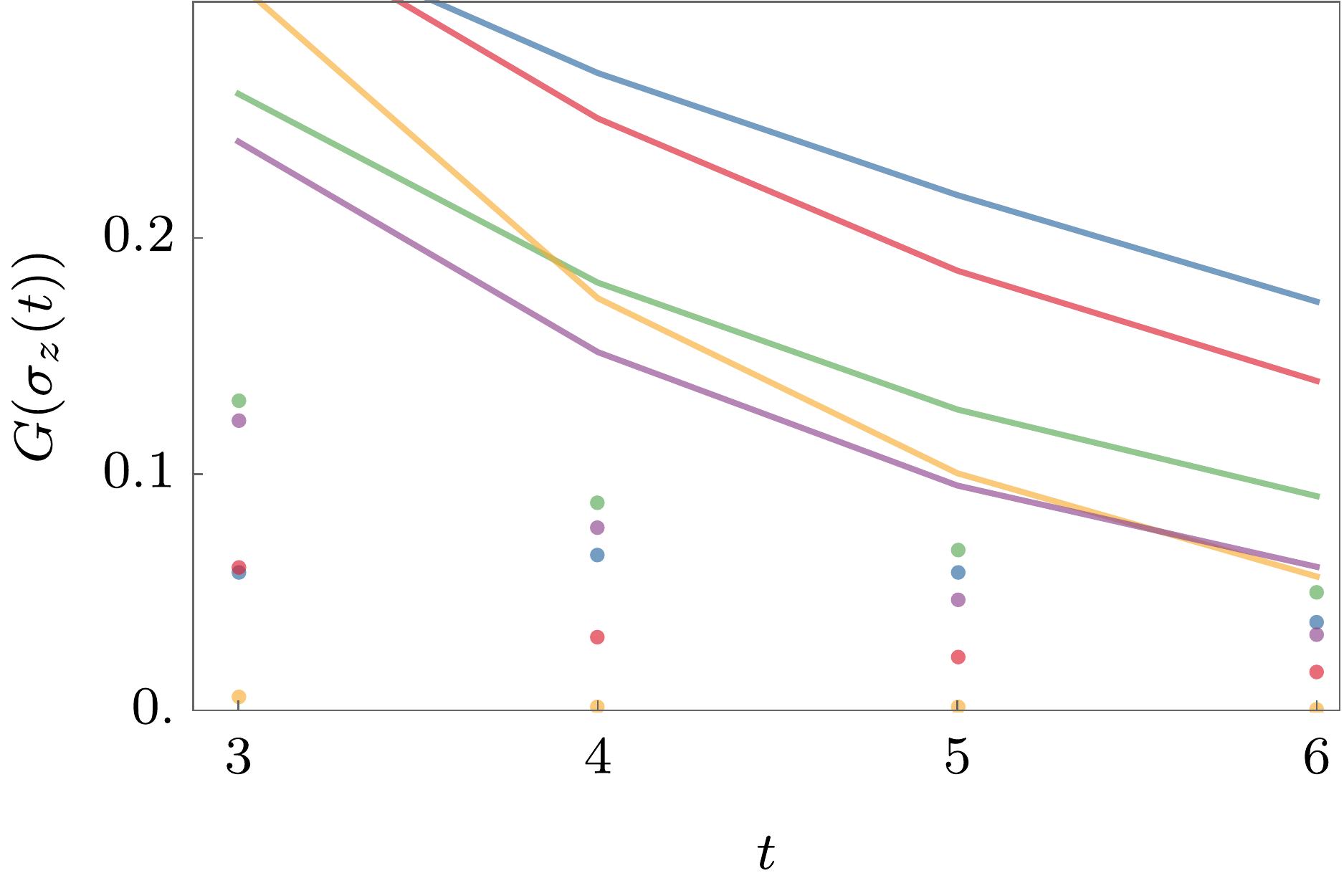}
    \caption{Comparison of the bound from Eq.~\eqref{eq:ReyniBound} (solid lines) and $G(\sigma_z(t)$ (points) for five different evolutions (different colors), which are given by a clean brick-wall quantum circuit made of the same two qubit gate. The five gates were chosen to be Haar random. The bipartition $A^{(\prime)}:\bar{A}^{(\prime)}$ is taken to be across half of the total system ($a=0$). 
    }
    \label{fig:bound}
\end{figure}

\subsection{Dual Unitary Circuits} \label{ap:dual}
As explained in the main text, dual unitarity is the extra condition on a brickwork circuit, that each unitary brick is unitary in both time \emph{and space} directions. Graphically, this means that in addition to the graphical rules Eqs.~\eqref{eq:unitarity}-\eqref{eq:unitarity1}, we also have that 
\begin{equation}
    \tr_{13}[U \otimes U^*]=\includegraphics[scale=1.5, valign=c]{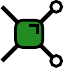}=\includegraphics[scale=1.5, valign=c]{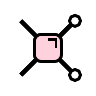}=\includegraphics[scale=1.5, valign=c]{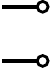},\label{eq:du_graph}
\end{equation}
as well as a spatial analogue of Eq.~\eqref{eq:unitarity1} (which is not relevant for the present proof, but is relevant for the results we use from Ref.~\cite{Kos2020}). Using this, we can now prove our main theorem of this section.

\du*
\begin{proof}
Directly from Eq.~\eqref{eq:nu_graph} we have that 
\begin{equation}
    \bra{\phi^+} \nu_A(t) \ket{\phi^+}=\includegraphics[scale=1.5, valign=c]{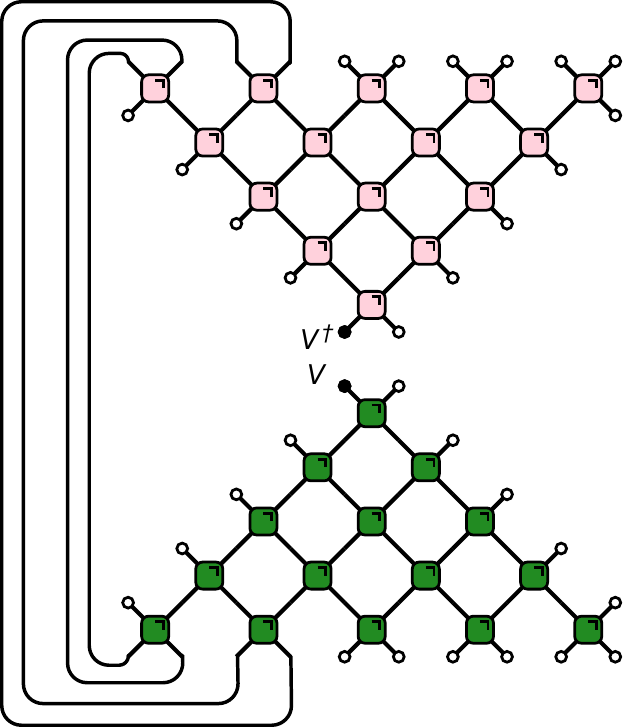}.
\end{equation}
where wlog we have taken $V$ to be local on the site $y=0$ (as in Fig.~\ref{fig:brick}). We have here first chosen the case where $\mc{H}_{\bar{A}^{(\prime)}}$ is on the left - i.e. that the right light cone of $V_t$ ends up in $\mc{H}_{A^{(\prime)}}$, and so $t \in \ell_A$. In the above diagram, the boundary site between $\mc{H}_{A^{(\prime)}}$ and $\mc{H}_{\bar{A}^{(\prime)}}$ is $a=-1/2$. A direct application of the graphical rules \eqref{eq:unitarity}, \eqref{eq:unitarity1} and \eqref{eq:du_graph} leads to
\begin{equation}
    \bra{\phi^+} \nu_A(t) \ket{\phi^+}|_{t \in \ell_A} =\includegraphics[scale=1.5, valign=c]{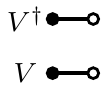}=|\tr[V]|^2 =0 \label{eq:dual_end},
\end{equation}
as $V$ is taken to be traceless. This holds in general when $t \in \ell_A$. If $V$ was instead on site $y=1$, the condition would change to $t \in \ell_A$ (left light cone in $A$), and otherwise the results are the same. Therefore wlog we just take $y=0$ here and in the rest of this work. It can be easily seen that this calculation holds for $y=0$ and arbitrary $a$ and $t$, as long as the right lightcone of $V_t$ ends up in $\mc{H}_{A^{(\prime)}}$. Substituting this result into Eq.~\eqref{eq:mainThm}, this completes the proof for the first condition of Eq.~\eqref{eq:G_du}, when $t\in \ell_A$.

The less trivial, complement situation is when the right light-cone of $V$ ends up in the traced-over region $\bar{A}^{(\prime)}$, $-t\in \ell_{A}$. This corresponds to, for example, 
\begin{equation}
    \bra{\phi^+} \nu_A(t) \ket{\phi^+}|_{V \in \mc{B}(\mc{H}_{\mathrm{even}})} =\includegraphics[scale=1.5, valign=c]{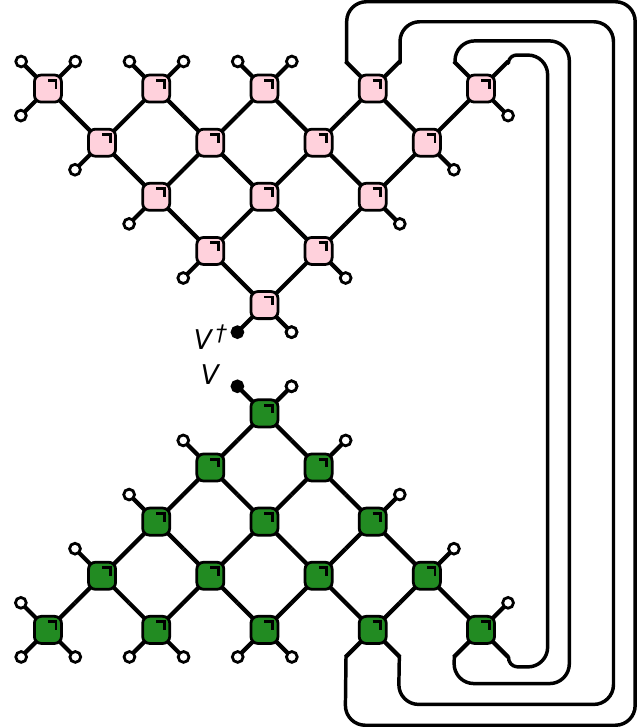} \, ,
\end{equation}
where now in this example $a=1$.
Again applying the graphical rules of dual and standard unitarity (Eqs.~\eqref{eq:unitarity}, \eqref{eq:unitarity1} and \eqref{eq:du_graph}), we arrive at the expression
\begin{equation}
    \bra{\phi^+} \nu_A(t) \ket{\phi^+}|_{t \in \ell_{\bar{A}}} =\includegraphics[scale=1.5, valign=c]{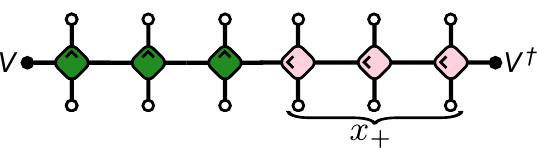}, \label{eq:Gfinal_du}
\end{equation}
recalling that $x_{\pm}:=t\pm a$. Graphically, one has that (defined algebraically below Eq.~\eqref{eq:G_du})
\begin{equation}
    \mc{M}_- := \includegraphics[scale=1.5, valign=c]{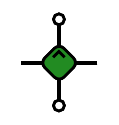} \quad \text{and} \quad \mc{M}_+ :=\includegraphics[scale=1.5, valign=c]{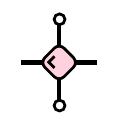}. \label{eq:Mdef}
\end{equation}
Again, this generalizes directly for arbitrary $a$ and $t$. This therefore completes the proof of Eq.~\eqref{eq:G_du} and therefore Theorem~\ref{thm:du}. 

We emphasize that we arrive at this expression only through the assumption of dual unitarity.
In particular, no \emph{completely chaotic} assumption from Refs.~\cite{Kos2020,Claeys2020} is needed so far. Further, this expression \eqref{eq:Gfinal_du} is efficient to compute numerically in any dual unitary circuit, with the matrices $\mc{M}_\pm$ governing the behavior of two point spatiotemporal correlation functions~\cite{Bertini2019exact}.

Interestingly, we will now see that the completely chaotic property results in an equivalence of chaos and scrambling. To make this explicit, we compare the above results with previous results for LOE. In particular, in Ref.~\cite{Kos2020} the LOE {R\'enyi} entropies for dual unitary circuits in the asymptotic light-cone limits $x_{\pm} \to \infty$ are computed. This is achieved for a completely chaotic dual unitary circuit. 

The completely chaotic (called `maximally chaotic' in Ref.~\cite{Claeys2020}) property corresponds to the assumption that the eigenvectors with eigenvalue one of transfer matrices $\mc{T}_s[U]$ (defined in Eq.~\eqref{eq:tranMatrix}), and an orthogonal transfer matrix (not defined explicitly here, but this is the transfer matrix in the `orthogonal lightcone direction'; see Eqs.~(37)-(38) in Ref.~\cite{Kos2020}) are limited to a minimal set. In particular, these minimal eigenvectors can be expressed in terms of the rainbow state \eqref{eq:rainbow}, and always trivially have eigenvalue one due to the dual unitary properties. See Eqs.~(49)-(51) of Ref.~\cite{Kos2020} for an explicit definition and derivation. We also note that in Ref.~\cite{Kos2020}, a numerical analysis reveals that the completely chaotic property is highly typical. What this means is that within the set of randomly chosen dual unitary circuits, elements that are not completely chaotic (likely) form a measure zero subset. It would be interesting to prove this claim analytically.

The precise scaling of the 2-{R\'enyi} entropy from Ref.~\cite{Kos2020} directly implies that for general dual unitary circuits, 
\begin{align}
    &\lim_{x_- \to \infty} \left(\exp[-\frac{1}{2} S^{(2)} (V_t )] \right) \geq \bra{V_t} \mc{M}_+^{x_+} \mc{M}_-^{x_+} \ket{V_t} \label{eq:KosResult}
\end{align}
In particular, Eq.~(80) from Ref.~\cite{Kos2020} is equal to our expression for $\braket{\phi^+ |\nu_A(t) |\phi^+}$ in Eq.~\eqref{eq:G_du} (note that our expressions for $x_\pm$ in this paper are reversed compared to \cite{Kos2020}, due to labeling conventions). Then for $G(V_t)$ in the limit $x_- \to \infty$, the multiplicative factor and additive constant go to one and zero respectively, as $d_A \to \infty $. 
The equality holds under the completely chaotic assumption and for $|\lambda| \geq d^{-1/2}$,  with the consecutive limits $x_+ \to \infty$ and $x_- \to \infty$ (as also necessarily $d_A\to \infty$). Here $|\lambda|$ is the largest non-trivial eigenvalue of $\mathcal{M}_-$. The more general Eq.~\eqref{eq:KosResult} also follows directly from Ref.~\cite{Kos2020}, where we notice that the expression for purity can contain additional positive terms for non completely chaotic examples, and in finite time, resulting in an inequality. 

We therefore arrive at Eq.~\eqref{eq:asymp_du} and therefore concluding the proof of Theorem~\ref{thm:completeChaos}. 

\end{proof}

\subsection{Scrambling Without Chaos in an Integrable Floquet Model} \label{ap:XXZ}
\scrambNoChaos*

\begin{proof} Expanding the single-site unitary qubit operator $V$ in the Pauli basis, we have that 
    \begin{align}
        V&=a_{\id} \id + a_x \sigma_X+ a_y \sigma_Y + a_z \sigma_Z   \nn \\
        &=a_x \sigma_X+ a_y \sigma_Y + a_z \sigma_Z \label{eq:pauliExp}
    \end{align}
     as $\tr[V]=0 \iff a_{\id}=0$. As we consider the normalized Choi state $\ket{V}$, we also have that 
     \begin{equation}
         a_z^2=1-(a_x^2+a_y^2). \label{eq:az}
     \end{equation}
     From Theorem~\ref{thm:du} we know that in dual unitary circuits, 
     \begin{equation}
         G(V_t)= \frac{1}{d_A^2-1}\left( d_A^2 \bra{V} \mc{M}_+^{x_+} \mc{M}_-^{x_+} \ket{V} - 1\right). \label{eq:G_du1}
     \end{equation}
     Now, using the full classification of dual unitary circuits for qubits~\cite{Bertini2019exact}, we know that $\mc{M}_{\pm}$ for the XXZ model for qubits takes the simple form in the Pauli basis
     \begin{equation}
         \mc{M}_\pm = \mathrm{diag}(1,\sin(2J),\sin(2J),1). \label{eq:xxzM}
     \end{equation}
     Substituting Eqs.~\eqref{eq:pauliExp} and \eqref{eq:xxzM} into Eq.~\eqref{eq:G_du1}, evaluating in the Pauli basis
     \begin{align}
          G(V_t)_{\mathrm{XXZ}}&=\frac{1}{d_A^2-1}\left(d_A^2((a_x^2+a_y^2)\sin(2J)^{2(x_-)}+a_z^2)-1 \right) \nn \\
          &=\frac{1}{d_A^2-1}\left(d_A^2((a_x^2+a_y^2)\sin(2J)^{2(x_-)}-(a_x^2+a_y^2))\right)+1,
     \end{align}
     where we have used Eq.~\eqref{eq:az}. 
     This corresponds to the exponential behavior decay with $(x_-)$, as $\sin(2J)<1$ given that $J\neq \pi/4$. For clarity of notation, we arrive at Eq.~\eqref{eq:xxz_result1} by setting
     \begin{align}
         \alpha&=\ln( \frac{1}{\sin(2J)}), \text{ and} \nn\\
         \beta&=d_A^2(a_x^2+a_y^2)/(d_A^2-1). \label{eq:AlphaBeta}
     \end{align}
     From this, for $a_z=0$ and taking the limit $(x_-) \to \infty$ we arrive at
    \begin{equation}
        \underset{(x_-) \to \infty}{\lim} \left(G(V_t)_{\mathrm{XXZ}}|_{a_z=0} \right)=\frac{-1}{d_A^2-1}.
    \end{equation}
\end{proof}

\section{Operator Free Generalization} \label{ap:OTOT}
Here we argue that a operator-free generalization of the OTOC and of the LOE are related in a straightforward manner, generalizing the main results of the Letter. These CP operators encode all possible OTOCs and local Heisenberg operator Choi states (from which to compute the LOE), analogous to defining a density matrix in quantum mechanics to encode all possible measurements which one could make. One can probe novel properties of this CP operator generalization, such as the conditional mutual information which allows one to distinguish between genuine quantum scrambling and decoherence (as in Ref.~\cite{Zonnios2022}).

The `out-of-time-order tensor' (OTOT) is defined as the CP operator $\ups^{\mathrm{OTOT}} \in \mc{H}_{B_i}\otimes \mc{H}_{A} \otimes \mc{H}_{B_f}$ such that~\cite{Zonnios2022}
\begin{equation}
    F(W,V_t)= \tr[ \ups^{\mathrm{OTOT}} (\mc{V} \otimes \mc{W} \otimes \mc{V^*} )].
\end{equation}
Full definitions and further details can be found in Ref.~\cite{Zonnios2022}. The space $B_{i/f}$ are the same spatial Hilbert spaces at the start and end of the OTOC protocol, which are technically indpendent spaces in the quantum combs formalism. We further take the initial state in the OTOT definition to be maximally mixed $\rho \sim \id$. Then for a choice of the operations $\mc{V}$ and $\mc{W}$, one gets exactly the usual OTOC $F(W,V_t)$ from the tensor $\ups^{\mathrm{OTOT}}$.

Similarly, we can generalize the LOE to a `local tensor entanglement' (LTE), for the CP operator $\ups^{\mathrm{LTE}} \in \mc{H}_{B^\prime}\otimes  \mc{H}_{S}$~\cite{Dowling2022},
\begin{equation} \label{eq:LOEgeneralization}
   \ups^{\mathrm{LTE}}:=\mc{U}_{S} \left( \frac{\id_{\bar{B}}}{d_{\bar{B}}}\otimes \ket{\phi^+}\!\bra{\phi^+}_{B^{(\prime)}} \right).
\end{equation}
As detailed in Ref.~\cite{Dowling2022}, if one projects with some choice of maximally entangled state onto the ancilla $B^\prime$, one gets exactly the time evolved Heisenberg operator. From this, the usual LOE can be computed. For example, for a qubit space $\mc{H}_B$, projecting onto the ancilla space $\mc{H}_{B^\prime}$ with the $\psi^+$ bell state results in,
\begin{equation}
    \braket{\psi^+| \ups^{\mathrm{LTE}}|\psi^+} = \ups_{S|x} = X_t,
\end{equation}
where $X_t$ is the time evolved Pauli-$X$ operator.

Our results from the main body can easily be framed in terms of these two objects.
\begin{obs} \label{obs:OTOT}
    The LTE and OTOT are related via
    \begin{equation}
        \ups_{B_i A_t B_f  } = \tr_{\bar{A}}[\ups_{N_iS_t} * \ups_{N_f S_t}^* ] 
    \end{equation}
    where the right hand side is the link product, giving a Hilbert Schmidt inner product on the complement to the probe space, $\mc{H}_{\bar{A}}$, and a tensor product on $\mc{H}_{{A}}$~\cite{Chiribella2009physrevA} (see below in Eq.~\eqref{eq:otot}). Note that $B_{i/f} \equiv N_{i/f} $ due to projection with the maximally entangled state in the definition Eq.~\eqref{eq:LOEgeneralization}.
\end{obs}
 This observation is immediately apparent graphically, in that 
\begin{equation}
    \ups^{\mathrm{OTOT}} = \includegraphics[width=0.38\textwidth, valign=c]{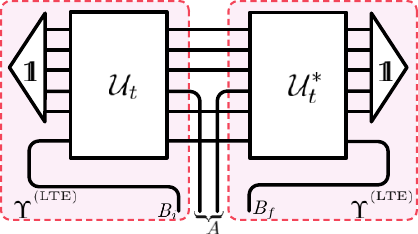}. \label{eq:otot}
\end{equation}
From this, all previous results will directly generalize to this operator-free setting. This is revealing of the close connection between the OTOC and LOE~\cite{Dowling2022}.

\end{document}